\documentclass[number,longtitle]{elsarticle}
\usepackage{latexsym,epsfig,amssymb,url}
\newcommand{\before}{\ensuremath{\prec}}
\begin{document}

\title{An Upper Bound on the Number of Circular Transpositions to Sort a Permutation}

\author{Anke van Zuylen\corref{cor1}}
\ead{anke@wm.edu}
\author{James Bieron}
\ead{jcbieron@email.wm.edu}
\author{Frans Schalekamp}
\ead{frans@wm.edu}
\author{Gexin Yu\fnref{fn1}}
\ead{gyu@wm.com}
\cortext[cor1]{Corresponding author}
\fntext[fn1]{Research was supported in part by NSA grant H98230-12-1-0226. }
\address{Department of Mathematics, The College of William and Mary, Williamsburg, VA, 23185, USA}

\begin{abstract}
We consider the problem of upper bounding the number of circular transpositions needed to sort a permutation.
It is well known that any permutation can be sorted using at most $\frac {n(n-1)}2$ adjacent transpositions.
We show that, if we allow all adjacent transpositions, as well as the transposition that interchanges the element in position $1$ with the element in the
last position, then the number of transpositions needed is at most $\lfloor\frac {n^2}4\rfloor$. This answers an open question posed by Feng, Chitturi and Sudborough (2010).
\end{abstract}
\begin{keyword}
Permutations; Cayley graphs; diameter; sorting
\end{keyword}
\maketitle

\newtheorem{definition}{Definition}
\newtheorem{theorem}{Theorem}
\newtheorem{lemma}[theorem]{Lemma}
\newtheorem{corollary}[theorem]{Corollary}
\newtheorem{conjecture}[theorem]{Conjecture}
\newtheorem{proposition}[theorem]{Proposition}
\newtheorem{claim}[theorem]{Claim}
\newtheorem{observation}{Observation}
\newtheorem{fact}{Fact}
\newtheorem{invariant}{Invariant}
\newenvironment{proof}{\noindent{\bf Proof:} \hspace*{1em}}{
    \hspace*{\fill} $\Box$ \\}
\newenvironment{proof_claim}{\noindent{\it Proof:} \hspace*{1em}}{
    \hspace*{\fill} $\diamond$ \\}

\section{Introduction}
The problem of sorting $n$ numbers has been well studied under a variety of different constraints.  Any introductory computer science course will include a discussion of sorting, and probably will begin with the simplest sorting algorithm of all, bubble sort.  Bubble sort, which can be described as easily as ``if two adjacent elements are out of order relative to each other, swap them'' is provably optimal in number of transpositions needed if you are only allowed to swap adjacent elements.  It can be shown that the number of adjacent transpositions that is needed to sort a given list of numbers is equal to the number of ``inversions'' in the list, i.e., pairs of elements that are out of order (but not necessarily adjacent).  Thus, the $\frac{n(n-1)}{2}$ bound on the number of inversions in a permutation of length $n$ also serves as the bound on the number of adjacent transpositions needed to sort said permutations.

Suppose we drop the requirement that the elements we swap be adjacent.  If we allow any two elements to be swapped, then it is known that $(n-1)$ transpositions will be needed in the worst case.  The algorithm here could be, ``swap the largest number into its proper place, and then sort the remaining $(n-1)$ elements recursively.''  There is no equivalent of ``inversions'' in this case, because we can always move at least one element all the way to its proper place; we do not need a local quantity telling us if two adjacent elements are out of order.

In this paper, we consider the properties of a third set of transpositions to be used for sorting a permutation: the set of all cyclically adjacent transpositions. By cyclically adjacent transpositions, we mean all normal adjacent transpositions plus an additional transposition: the one that swaps the first and last elements in the permutation.  This is a natural extension of the normal case of adjacent transpositions.  Consider if you were sorting elements in a data structure that was cyclic, instead of linear in nature.  Clearly, bubble sort is an option in this case; we still have all adjacent transpositions, so we could just ignore the added cyclic transposition and proceed as before.  As one would expect, this is far from an optimal algorithm.  The problem of designing an algorithm to find an optimal sequence of cyclically adjacent transpositions to sort a permutation was examined and solved by Jerrum~\cite{Jerrum85}.  However, while an optimal algorithm is given, the question of an upper bound on the number of swaps that will be required in the worst case remained unanswered.  In 2010, this problem was posed by Feng, Chitturi and Sudborough~\cite{FengCS10} in a paper in which they prove that $\lfloor \frac{n^2}{4} \rfloor$ is a lower bound, and they conjecture that this bound is tight.  Here, we prove that $\lfloor \frac{n^2}{4} \rfloor$ is indeed an upper bound on the number of cyclically adjacent transpositions needed to sort any permutation of length $n$, thus resolving the open question of Feng et al.~\cite{FengCS10}.
\section{Preliminaries}
We now introduce the notation we will use throughout this note.

Let $\pi$ be a permutation of $\{1, \ldots, n\}$. We will refer to $\pi(i)$ as the position of element $i$.
If $\pi(i)=p$, we have $\pi^{-1}(p)=i$, i.e., $\pi^{-1}(p)$ gives the element that is in position $p$.
We will sometimes write $\pi^{-1}$ as the ordered sequence $(\pi^{-1}(1), \pi^{-1}(2), \ldots, \pi^{-1}(n))$.
In the following, we will use $i,j,k$ when we want to refer to an element in $\{1, \ldots, n\}$ and $p,q,r$ to refer to a position in $\{1, \ldots, n\}$.

Given a permutation $\pi$ and $p,q\in \{1, \ldots, n\}$, applying the transposition $(p,q)$ to $\pi$ means that we ``swap'' the elements in positions $p$ and $q$ to obtain a new permutation $\tilde \pi$, where $\tilde\pi^{-1}(p)=\pi^{-1}(q), \tilde \pi^{-1}(q)=\pi^{-1}(p)$, and $\tilde \pi^{-1}(r)=\pi^{-1}(r)$ for all $r\in \{1, \ldots, n\}\backslash\{p,q\}$.

We say a transposition $(p,q)$ is adjacent if $q=p+1$, and we say a transposition $(p,q)$ is cyclically adjacent if either $q=p+1$, or $p=n$ and $q=1$.
For ease of exposition, we will use $(p,p+1)$ to denote a cyclically adjacent transposition (i.e., the fact that $p+1$ is taken modulo $n$ is implicit).
We will sometimes refer to a transposition $(p, p+1)$  when applied to $\pi$ with $\pi^{-1}(p)=i, \pi^{-1}(p+1)=j$, as the {\it swap} $(i,j)$ of elements $i$ and $j$.
We remark that a swap is denoted as an ordered pair, where the first element moves ``in clockwise direction'', i.e., from position $p$ to $p+1$ and the second element moves in ``counterclockwise direction'', i.e., from position $p+1$ to $p$. 

We say that $Q = ( q_1, q_2, \ldots, q_m )$ is a sequence of cyclically adjacent swaps for $\pi$, if for every $i=1,2,\ldots,m$ we have that $q_i$ is a cyclically adjacent swap for the permutation that results from performing the swaps $q_1, q_2, \ldots, q_{i-1}$ in order on permutation $\pi$.

We say a permutation $\pi$ is sorted by a sequence of transpositions, if we obtain the identity after the sequence of transpositions is applied to $\pi$.
It is well known that any permutation $\pi$ can be sorted by at most $\frac{n(n-1)}2$ adjacent transpositions.

In this note, we will show the following theorem. This answers an open question of Feng, Chitturi and Sudborough~\cite{FengCS10}.
\begin{theorem}
Given any permutation $\pi$ of $\{1,\ldots,n\}$, there exists a sequence of at most $\frac{n^2}{4}$ cyclically adjacent transpositions to sort $\pi$.
\end{theorem}

To prove the theorem, we begin by reviewing results by Jerrum~\cite{Jerrum85}.
Given a sequence of cyclically adjacent transpositions that sort $\pi$,
we consider the corresponding sequence of swaps of elements. For this sequence of swaps, we let $c(i,j)$ be the number of times swap $(i,j)$ occurs minus the number of times swap $(j,i)$ occurs. We define the net clockwise displacement for element $i$ as $d(i)=\sum_{j\neq i} c(i,j)$.
Then we have that
\begin{equation}\sum_{i=1}^n d(i)=0,\label{eq:feas1}\end{equation}
since $\sum_i d(i) = \sum_i \sum_{j\neq i} c(i,j) = \sum_{(i,j): c(i,j)>0} c(i,j) + \sum_{(j,i):c(j,i)<0} c(j,i)$$=
\sum_{(i,j): c(i,j)>0} c(i,j) - \sum_{(i,j):c(i,j)>0} c(i,j)=0$, where the penultimate equality uses the fact that $c(i,j)=-c(j,i)$.
Since the sequence of transpositions sorts $\pi$, it must be the case that
\begin{equation}
\pi(i)+d(i) \equiv i \pmod n,\mbox{ for every $i\in\{1, \ldots, n\}$.}\label{eq:feas2}
\end{equation}

Jerrum's key result is a characterization of the net displacement vector $d$ of a minimum length sequence of cyclically adjacent transpositions to sort $\pi$.
We first show how, given a displacement vector $d$ that satisfies (\ref{eq:feas1}) and (\ref{eq:feas2}), we can find a sequence of cyclically adjacent transpositions that sort $\pi$ and have net displacement vector $d$.
We then give the expression given by Jerrum for $c(i,j)$, the net numer of times swap $(i,j)$ occurs in this sequence, as a function of $\pi$ and $d$.
Finally, we give Jerrum's main result which characterizes the displacement vector $d$ that corresponds to the minimum length sequence of transpositions that sort $\pi$.

\begin{lemma}\label{lem:sort}
Given a displacement vector $d$ that satisfies (\ref{eq:feas1}) and (\ref{eq:feas2}) with respect to some permutation $\pi$, a sequence of cyclically adjacent swaps that sort $\pi$ and has net displacements given by $d$ is found by repeatedly swapping cyclically adjacent elements $(i,j)$ such that $d(i)>d(j)$, and decreasing $d(i)$ by 1 and increasing $d(j)$ by 1.
\end{lemma}
\begin{proof}
Note that cyclically adjacent elements $(i,j)$ such that $d(i)>d(j)$ exist unless $d(i)=0$ for every $i=1, \ldots, n$, since $d$ satisfies $\sum_{i=1}^n d(i)=0$.
After executing the swap, we decrease $d(i)$ by 1 and we increase $d(j)$ by 1. Let $\tilde d$ be the new displacement vector. It is easily verified that the new permutation $\tilde\pi$ and displacement vector $\tilde d$ obtained after executing the swap satisfy (\ref{eq:feas1}) and (\ref{eq:feas2}). Hence, if this process terminates, then it will result in a sequence of cyclically adjacent swaps to sort $\pi$ with net displacement vector $d$.

We now argue that $\sum_k (\tilde d(k))^2 < \sum_k (d(k))^2$, which implies that this process does indeed terminate.

Note that 
\begin{eqnarray*}
\sum_k (\tilde d(k))^2 - \sum_k (d(k))^2 &=& (d(i)-1)^2 + (d(j)+1)^2 - \left((d(i))^2 +(d(j))^2\right) \\
&=& -2d(i)+2d(j) + 2.\end{eqnarray*}
Now, note that, since $\pi(i)=p$ and $\pi(j)\equiv p+1 \pmod n$, then $d(i)>d(j)$ implies that $d(i)\ge d(j)+2$, since $\pi$ and $d$ satisfy (\ref{eq:feas1}).
Hence $-2d(i)+2d(j)+2\le -2$.
\end{proof}

We note that this lemma generalizes sorting $\pi$ by using only adjacent transpositions (i.e., bubble sort) in a natural way: In that case, we take $d(i) = i - \pi(i),$ for every $i\in\{1, \ldots, n\}$. Now, let $i=\pi^{-1}(p)$ and $j=\pi^{-1}(p+1)$. Then $d(i) = i-p$ and $d(j) = j-(p+1)$, and hence $d(i)>d(j)$, implies that $i>j-1$, and hence, $i>j$, since $i\neq j$.  So in this case, the algorithm in Lemma~\ref{lem:sort} is simply the bubble sort algorithm, in which we swap adjacent elements $(i,j)$ if $i>j$.

We now give two results that were shown by Jerrum. The expression we use in the next lemma gives an expression for $c(i,j)$, the net number of times swap $(i,j)$ occurs in the sequence resulting from Lemma~\ref{lem:sort}. The expression is essentially the same as the expression derived on page 283 of~\cite{Jerrum85}. For completeness, we give a proof in the appendix.
\begin{lemma}[Jerrum\cite{Jerrum85}]\label{lem:Jerrum1}
Given a displacement vector $d$ that satisfies (\ref{eq:feas1}) and (\ref{eq:feas2}) with respect to some permutation $\pi$, any sequence of cyclically adjacent transpositions that sorts $\pi$ and has net displacement vector $d$ has
\[c(i,j) = \left\{\begin{array}{ll}1+\max\{m: \pi(i)+d(i)> \pi(j)+d(j)+mn \}&\mbox{
if $\pi(i)<\pi(j)$,}\\
\max\{m: \pi(i)+d(i)> \pi(j)+d(j)+mn\} &\mbox{
if $\pi(i)>\pi(j)$}.
\end{array}\right.\]
\end{lemma}

Jerrum's main result is a characterization of the net displacement vector $d$ of the minimum length sequence of cyclically adjacent transpositions for a given permutation $\pi$.
The following theorem summarizes the results in Corollary 3.7 and Theorem 3.9 of~\cite{Jerrum85}.
\begin{theorem}[Jerrum~\cite{Jerrum85}]
A sequence of cyclically adjacent transpositions that sort permutation $\pi$ is of minimum length if and only if 
each pair of elements is swapped at most once, and the net displacement vector $d$
satisfies
\begin{equation}d(i)-d(j)\le n \mbox{ for all }i,j\in\{1, \ldots, n\}.\label{eq:opt}\end{equation}
\end{theorem}
We omit the proof since the results in the next section only rely on the fact that for any permutation $\pi$, there exists a sequence of cyclically adjacent transpositions that sort $\pi$ and for which the net displacement vector $d$ satisfies (\ref{eq:opt}).

To find this sequence, we initialize $d(i)=i-\pi(i)$ for $i=1, \ldots, n$. Note that $|d(i)|\le n$ for every $i$.

Now, if $d(i)-d(j) >n$, then $d(i)>0$ and $d(j)<0$. If we subtract $n$ from $d(i)$ and add $n$ to $d(j)$, we obtain a new valid displacement vector $d'$, which has $d'(i)=d(i)-n <0$ and $d'(j)=d(j)+n >0$.  Therefore, $\sum_k |d'(k)| = \sum_{k} |d(k)|  - |d(i)|+|d(i)-n|-|d(j)|+|d(j)+n|=\sum_k |d(k)| + 2n - 2(d(i)-d(j)) < \sum_k |d(k)|$. Hence, this process will terminate. We can then use Lemma~\ref{lem:sort} to find the corresponding sequence of cyclically adjacent transpositions.

\section{An upper bound on the number of cyclically adjacent transpositions to sort $\pi$.}
By the results from the previous section, we know that for any permutation $\pi$ of $\{1, \ldots, n\}$, there exists a sequence of cyclically adjacent transpositions with net displacement vector $d$ which satisfies (\ref{eq:opt}) that sorts $\pi$.
If we apply a cyclically adjacent swap $(i,j)$ with $d(i)>d(j)$, then $d(i)$ is decreased by one, and $d(j)$ is increased by one. 
Hence, $\frac12 \sum_{i=1}^n |d(i)|$, where $d$ satisfies (\ref{eq:feas1}), (\ref{eq:feas2}) and (\ref{eq:opt}), is a {\it lower bound} on the number of cyclically adjacent transpositions needed to sort $\pi$.
The maximum value this lower bound can take is $\frac {n^2}{4}$, since $d$ satisfies (\ref{eq:opt}). It was shown in~\cite{FengCS10} that this bound is tight for the permutation $\pi^{-1}=(\frac n2+1, \frac n2+2, \ldots, n-1, n, 1, 2, \ldots, \frac n2)$, where $n$ is even.

One might conjecture that there always exists a swap such that
$\frac12\sum_{i=1}^n |d(i)|$ decreases by one, which would prove that $\frac{n^2}{4}$ is also an {\it upper bound} on the number of cyclically adjacent transpositions needed to sort any permutation $\pi$.
However, this is only if there exists an adjacent swap $(i,j)$ where $d(i)>0$ and $d(j)<0$. The following example shows that such a swap does not always exist: let $\pi^{-1}=(3, 2, 1, 4)$. Then $d=(2,0,-2,0)$ satisfies (\ref{eq:feas1}), (\ref{eq:feas2}) and (\ref{eq:opt}), but any cyclically adjacent swap does not decrease $\sum_{i=1}^n |d(i)|$. 
In this section, we use different techniques to show that the conclusion does hold that at most $\frac{n^2}{4}$ cyclically adjacent transpositions suffice to sort any permutation $\pi$.

We begin by stating two auxiliary lemmas. To maintain the flow of the argument, we defer their proofs until later in this section.
We slightly generalize the notion of permutation to be a bijection of any set $S$ of positive integers. Note that the net displacement vector of a sequence of cyclically adjacent swaps is still well defined as $d(i) = \sum_{j \neq i} c( i, j )$. The following lemma follows from Lemma~\ref{lem:Jerrum1}.

\begin{lemma}\label{lem:prop}
Let $\pi$ be a permutation of $S$, and let $Q = ( q_1, q_2, \ldots, q_m )$ be a sequence of cyclically adjacent swaps for $\pi$, with net displacement vector $d$ that satisfies (\ref{eq:opt}). Then for any two distinct elements $i,j\in S$,
\begin{enumerate}
\item[(a)]
$d(i)\ge d(j)$ implies $0\le c(i,j)\le 1$.
\item[(b)]
$d(i)=d(j)$ implies $c(i,j)=0$.
\item[(c)]
$d(i)-d(j)=n$ implies $c(i,j)=1$.
\end{enumerate}
\end{lemma}

We will prove our main result by induction, and in order to use the inductive hypothesis, we will remove some element $k$ from $S$. We now define what we mean by a permutation corresponding to $\pi$ restricted to $S \setminus \{ k \}$.
First of all, we define the relationship ``$i$ is directly before $j$ in a permutation $\pi$'' if either $\pi( i ) < \pi( j )$ and there is no $\ell \in S$ so that $\pi( i ) < \pi( \ell ) < \pi( j )$, or $\pi( i ) = \max S$ and $\pi( j ) = \min S$, where $\max S$ and $\min S$ are the largest and smallest integer in $S$ respectively. We denote this relationship by $i \before j$.

Given a permutation $\pi$ of $S$ and a permutation $\pi': S\setminus\{k\} \to S\setminus\{k\}$. We say $\pi'$ is a permutation corresponding to $\pi$ restricted to $S \setminus \{ k \}$ if $\pi'$ preserves the $\before$-relationship, i.e. $i \before j$ in $\pi'$ for $i,j \in S'$ if $i \before j$ in $\pi$, or $i \prec k \prec j$ in $\pi$. We note that there are $|S|-1$ distinct permutations corresponding to $\pi$ restricted to $S \setminus \{ k \}$, but this will not be important for our purposes.

\begin{lemma}\label{lem:reduce}
Let $\pi$ be a permutation of $S$, and let $Q = ( q_1, q_2, \ldots, q_m )$ be a sequence of cyclically adjacent swaps for $\pi$, resulting in permutation $\sigma$. Let $Q'$ be the sequence of swaps, where all swaps involving element $k$ are deleted, and let $\pi'$ be any permutation corresponding to $\pi$ restricted to $S \setminus \{ k \}$.  Then $Q'$ is a sequence of cyclically adjacent swaps for $\pi'$, and performing $Q'$ on $\pi'$ results in a permutation corresponding to $\sigma$ restricted to $S \setminus \{ k \}$.
\end{lemma}

We now rephrase our main theorem, and use the two previous lemmas to prove it. 
By the results of the previous section, we have that for any permutation $\pi$ there exists a minimum length sequence of cyclically adjacent transpositions that sorts $\pi$, so that every pair of elements is swapped at most once, and the net displacement vector satisfies (\ref{eq:opt}). It is therefore sufficient to prove that a sequence of cyclically adjacent swaps with the properties that each pair of elements is swapped at most once, and the net displacement vector satisfies (\ref{eq:opt}), has length at most $\frac {n^2}4$, where $n$ is the number of elements.

\begin{theorem}\label{lemma:main}
Consider a sequence of cyclically adjacent swaps for a permutation $\pi$ of a set of $n$ elements, where each pair of elements is swapped at most once, and for which the net displacement vector $d$ satisfies (\ref{eq:opt}). Then the sequence consists of at most $\frac{n^2}4$ swaps.
\end{theorem}
\begin{proof}
We prove the lemma by induction on $n$.
If $n=2$, then the lemma is clearly true, as in this case there is only one pair of elements, and this pair can be swapped at most once.

Now, assume the lemma is true for $n'=n-1$. Consider a sequence of cyclically adjacent swaps for a permutation $\pi$ of a set $S$ of $n$ elements that satisfies the conditions of the lemma.

Let $d_{\max} = \max_i d(i)$ and let $d_{\min} = \min_i d(i)$. Note that $d_{\max}-d_{\min} \le n$ by (\ref{eq:opt}), and hence either $d_{\max} \le \frac n2$ or $-d_{\min} \le \frac n2$.
In the first case, let $k$ be an element such that $d(k)=d_{\max}$; in the second case, we let $k$ be such that $d(k)=d_{\min}$.

In order to use the inductive hypothesis, we remove element $k$ from $S$. 
We let $Q'$ be the sequence of swaps, where all swaps involving element $k$ are deleted, and we 
let $\pi'$ be any permutation corresponding to $\pi$ restricted to $S \setminus \{ k \}$.  By Lemma~\ref{lem:reduce}, $Q'$ is a sequence of cyclically adjacent swaps for $\pi'$.
For the sequence $Q'$, let $d'(i)$ be the net clockwise displacement of element $i$ for any $i\in S \setminus \{ k \}$. Note that $d'( i ) = d( i ) - c( i,k )$.
Below, we will show that for any $i,j\in S \setminus \{ k \}$, we have $d'(i)-d'(j) \le n-1$.
Hence, $Q'$ corresponds to a sequence of cyclically adjacent swaps for a permutation of $S\setminus\{k\}$ with net displacement vector satisfying (\ref{eq:opt}) in which each pair of elements is swapped at most once. By the inductive hypothesis, $Q'$ can have at most $\frac{(n-1)^2}4$ swaps.
In addition, we will show that $k$ is involved in exactly $|d(k)|$ swaps. Since $|d(k)|\le \frac n2$, we conclude that the total number of swaps in the original sequence is at most $\frac{(n-1)^2}4+\frac n 2=\frac {n^2+1}4$. Since the number of swaps is integer, it can thus be at most $\lfloor \frac {n^2+1}4\rfloor = \frac{n^2}4$.

To prove the two claims, we use Lemma~\ref{lem:prop}.
First, suppose that $d(k)=d_{\max}$. Then
\begin{eqnarray*}
d'(i)-d'(j) &=& d(i)-c(i,k)-(d(j)-c(j,k))\\
&=& d(i)+c(k,i)-d(j)-c(k,j)\\
&\le& d(k)-d(j)-c(k,j)\\
&\le& n-1.
\end{eqnarray*}
The first inequality uses the fact that $d(k)\ge d(i)$,  so that $c(k,i)\le 1$ by property (a) of Lemma~\ref{lem:prop}, and $c(k,i)=0$ if $d(k)=d(i)$ by property (b).
The second inequality uses the fact that $d(k)-d(j)\le n$ by (\ref{eq:opt}), plus the fact that $c(k,j)\ge0$ by property (a), and $c(k,j)=1$ if $d(k)-d(j)=n$ by property (c).

The proof when $d(k)=d_{\min}$ is similar, and is included for completeness. In this case, we write
\begin{eqnarray*}
d'(i)-d'(j) &=& d(i)-c(i,k)-(d(j)-c(j,k))\\
&\le& d(i)-c(i,k)-d(k)\\
&\le& n-1.
\end{eqnarray*}
The first inequality uses the fact that $-d(j)\le -d(k)$, by definition of $k$, and property (a) and (b) in Lemma~\ref{lem:prop}.
The second inequality uses the fact that $d(i)-d(k)\le n$ by (\ref{eq:opt}), plus property (a) and (c) from Lemma~\ref{lem:prop}.

Finally, note that the number of swaps in which $k$ is involved is $\sum_i |c(i,k)|$, and by (a) and the definition of $k$, this is exactly equal to $|d(k)|$.
\end{proof}

We conclude by giving the proofs of Lemma~\ref{lem:prop} and Lemma~\ref{lem:reduce}.\\

\newenvironment{proof_of}[1]{\noindent {\bf Proof of #1:}
    \hspace*{1em} }{\hspace*{\fill} $\Box$ }

\begin{proof_of}{Lemma~\ref{lem:prop}}

Let $\tau$ be the permutation of $S$ obtained by applying $Q$ to $\pi$.
Relabel the elements of $S$ with $\{1, \ldots, n\}$ so that $\tau$ is equal to the identity, and use the same relabeling on $\pi$ and $Q$.
Then, $\pi$ is a permutation of $\{1, \ldots, n\}$ and $Q$ sorts $\pi$. Hence, the net displacement vector $d$ corresponding to $Q$ and $\pi$ satisfy the conditions of Lemma~\ref{lem:Jerrum1}. For any $(i,j)$ we thus have
\[c(i,j) = \left\{\begin{array}{ll}1+\max\{m: \pi(i)+d(i)> \pi(j)+d(j)+mn \}&\mbox{
if $\pi(i)<\pi(j)$,}\\
\max\{m: \pi(i)+d(i)> \pi(j)+d(j)+mn\} &\mbox{
if $\pi(i)>\pi(j)$}.
\end{array}\right.\]

We take $i,j$ such that $d(i)\ge d(j)$. Note that (\ref{eq:opt}) implies that $d(j)\le d(i)\le d(j)+n$.

Suppose that $\pi(i)>\pi(j)$. Note that $\pi(i)<\pi(j)+n$, since $\pi$ is a permutation.
Hence \[\pi(j)+d(j)<\pi(i)+d(i)<\pi(j)+d(j)+2n,\] so
$c(i,j)=\max\{m: \pi(i)+d(i)> \pi(j)+d(j)+mn\}\in \{0,1\}$.
Moreover, if $d(i)=d(j)$, then we have $\pi(j)+d(j)<\pi(i)+d(i)<\pi(j)+d(j)+n$, so $c(i,j)=0$, and if
$d(i)=d(j)+n$, then $\pi(i)+d(i)>\pi(j)+d(j)+n$, so $c(i,j)= 1$. 

Similarly, if $\pi(i)<\pi(j)<\pi(i)+n$, then
\[\pi(j)+d(j)-n<\pi(i)+d(i)<\pi(j)+d(j)+n.\] 
Since
$c(i,j)=1+\max\{m:\pi(i)+d(i)> \pi(j)+d(j)+mn \}$ we get that $0\le c(i,j)\le 1$.
If $d(i)=d(j)$, then
$\pi(j)+d(j)-n<\pi(i)+d(i)<\pi(j)+d(j)$, so $c(i,j)=0$, and if
$d(i)=d(j)+n$, then $\pi(j)+d(j)<\pi(i)+d(i)$, so $c(i,j)= 1$. 
\end{proof_of}

\begin{proof_of}{Lemma~\ref{lem:reduce}}

We prove the lemma by induction on the number of swaps $m$.
For $m=0$ the claim is vacuously true. For general $m$, denote by $\tau$ the permutation that results from performing the first $m-1$ swaps in $Q$ on $\pi$, and $\tau'$ the permutation that results from performing the corresponding swaps in $Q'$ on $\pi'$. By the inductive hypothesis we know that $\tau'$ is a permutation corresponding to $\tau$ restricted to $S \setminus \{ k \}$. We now discern two cases.

(case 1) $k$ is not swapped by $q_m = (i,j)$, i.e. $i\neq k$ and $j\neq k$. We note that
$i \before j$ in $\tau$ because $q_m = (i,j)$ is a valid cyclically adjacent swap for $\tau$. Then $i \before j$ in $\tau'$ by the fact that $\tau'$ is a restricted permutation, so $q_m$ is also a valid cyclically adjacent swap for $\tau'$.

Performing one cyclically adjacent swap only changes the $\before$-relationship  for pairs of elements for which at least one element is in $\{i,j\}$. 
Let $\ell$ be so that $\ell \prec i$ in $\tau$. Then $\ell \prec j$ in the permutation that results after swapping $q_m=(i,j)$ in $\tau$. If $\ell \neq k$ then $\ell \prec i$ in $\tau'$ as well, and therefore $\ell \prec j$ in the permutation that results after performing swap $q$ on $\tau'$. If $\ell = k$, then let $\ell'$ be so that $\ell' \prec k$ in $\tau$, which means that $\ell' \prec i$ in $\tau'$. After performing $q_m$ on $\tau'$, we have $\ell' \prec j$.  

Checking the $\prec$-relations in the restricted permutation for the element right after $j$ in $\tau$ proceeds similarly.

(case 2) $k$ is an element that is swapped by $q_m = ( i, j )$, i.e. $k\in\{i,j\}$. We let $\tau'$ be any permutation corresponding to $\tau$ restricted to $S \setminus \{ k \}$. Let $\ell$ be the element in $\{i,j\}$ that is not equal to $k$. Let $a$ and $b$ be so that $a \prec i \prec j \prec b$ in $\tau$. Then $a \prec \ell \prec b$ in $\tau'$. Also, we know that $a \prec j \prec i \prec b$ in the permutation $\sigma$ obtained after applying $q_m$ to $\tau$. This means that a permutation corresponding to $\sigma$ restricted to $S \setminus \{ k \}$ will have $a \prec \ell \prec b$ as well. In other words, $\tau'$ is a permutation corresponding to $\sigma$ restricted to $S \setminus \{ k \}$, since it has the required $\before$-relationship between the elements.
\end{proof_of}

\section*{Acknowledgement}
The authors acknowledge Chi-Kwong Li for suggesting this research topic, for leading weekly research meetings, and for numerous helpful discussions.

\bibliographystyle{plain}

\begin{thebibliography}{1}

\bibitem{FengCS10}
Xuerong Feng, Bhadrachalam Chitturi, and Hal Sudborough.
\newblock Sorting circular permutations by bounded transpositions.
\newblock {\em Advances in Computational Biology: Advances in Experimental
  Medicine and Biology}, Volume 680:725--736, 2010.

\bibitem{Jerrum85}
Mark Jerrum.
\newblock The complexity of finding minimum-length generator sequences.
\newblock {\em Theor. Comput. Sci.}, 36:265--289, 1985.

\end{thebibliography}

\appendix

\section{Proof of Lemma~\ref{lem:Jerrum1}}
\begin{proof}
Consider a permutation $\pi$ and a net displacement vector $d$ that satisfies the conditions of the lemma.
First, we note that $c$ is skew symmetric, i.e., $c(i,j)=-c(j,i)$ for any $i\neq j$.
To see this, suppose without loss of generality that $\pi(i)<\pi(j)$, and let $c(i,j)=m$.
Then $\pi(i)+d(i)>\pi(j)+d(j)+(m-1)n$ and $\pi(i)+d(i) < p(j)+d(j) + mn$ (where the inequality is strict because of (\ref{eq:feas2})).
Therefore $\pi(j)+d(j) < \pi(i)+d(i) + (1-m)n$ and $\pi(j)+d(j) > \pi(i)+d(i)+(-m)n$, so $c(j,i)=-m$.

Applying a transposition $(p,p+1)$ results in a new permutation
$\tilde \pi$ and a new displacement vector $\tilde d$. Let $k=\pi^{-1}(p)$ and $\ell=\pi^{-1}(p+1)$, then $\tilde \pi$ and $\tilde d$ are given by
$\tilde \pi(i)=\pi(i), \tilde d(i)=d(i)$ for $i\neq k,\ell$ and $\tilde \pi(k)$ is $p+1$ if $p<n$ and $1$ if $p=n$, $\tilde d(k)=d(k)-1, \tilde \pi(\ell)=p, \tilde d(\ell)=d(\ell)+1$.
It is clear that $\tilde d$ satisfies the conditions of the lemma for $\tilde \pi$. Let
\[\tilde c(i,j) = \left\{\begin{array}{ll}1+\max\{m: \tilde \pi(i)+\tilde d(i)> \tilde \pi(j)+\tilde d(j)+mn \}&\mbox{
if $\tilde \pi(i)<\tilde \pi(j)$,}\\
\max\{m: \tilde \pi(i)+\tilde d(i)> \tilde \pi(j)+\tilde d(j)+mn\} &\mbox{
if $\tilde \pi(i)>\tilde \pi(j)$}.
\end{array}\right.\]
In order to prove the lemma, we need to show that swapping $(k,\ell)$ decreases $c(k,\ell)$ by one, i.e., $\tilde c(k,\ell)=c(k,\ell)-1$. By skew symmetry, this also implies that $\tilde c(\ell,k)=c(\ell,k)+1$.
In addition, we need to show that  for any pair $(i,j)\neq \{(k,\ell), (\ell,k)\}$, $c(i,j)=\tilde c(i,j)$.

We first consider a pair of elements $(i, j)\not \in\{ (k,\ell),(\ell,k)\}$.
If $p<n$, then $\tilde c(i,j)=c(i,j)$ because $\pi(i)+d(i)=\tilde \pi(i)+\tilde d(i)$ for every $i\in \{1, \ldots, n\}$ and the relative order of all pairs of elements, except $(k,\ell)$ is the same in $\pi$ and $\tilde\pi$.
If $\pi=n$, then the relative order of every pair containing $k$ or $\ell$ is changed, but it is easily verified that the fact that $\tilde \pi(k)+\tilde d(k)=\pi(k)+d(k)-n$ and $\tilde \pi(\ell)+\tilde d(\ell)=\pi(\ell)+d(\ell)+n$ implies that $c(i,j) = \tilde c(i,j)$ unless $(i, j)$ is $(k,\ell)$ or $(\ell,k)$.

We now consider the pair $(k,\ell)$, and show that $\tilde c(k,\ell) = c(k,\ell)-1$. If $p<n$, then $\pi(k)<\pi(\ell)$ and $\tilde \pi(k)>\tilde \pi(\ell)$. Also, $\tilde \pi(k)+\tilde d(k)=\pi(k)+d(k)$ and $\tilde \pi(\ell)+\tilde d(\ell)=\pi(\ell)+d(\ell)$. Hence, $\tilde c(k,\ell)=\max\{m: \pi(k)+d(k)>\pi(\ell)+d(\ell)+mn\}$ and $c(k,\ell)=1+\max\{m:\pi(k)+d(k)>\pi(\ell)+d(\ell)+mn\}$.
We thus have that $\tilde c(k,\ell)= c(k,\ell)-1$.
If $p=n$, then $\pi(k)>\pi(\ell)$, $\tilde \pi(k)<\tilde \pi(\ell)$, $\tilde \pi(k)+\tilde d(k)=\pi(k)+d(k)-n$ and $\tilde \pi(\ell)+\tilde d(\ell)=\pi(\ell)+d(\ell)+n$.
Therefore, $\tilde c(k,\ell)=1+\max\{m: \pi(k)+d(k)-n>\pi(\ell)+d(\ell)+n+mn\}=-1+\max\{m':\pi(k)+d(k)>\pi(\ell)+d(\ell)+m'n\}$ and $c(k,\ell)=\max\{m:\pi(k)+d(k)>\pi(\ell)+d(\ell)+mn\}$. Hence, we again have that $\tilde c(k,\ell)=c(k,\ell)-1$.
\end{proof}

\end{document}